\documentclass[11pt,a4paper]{article}

\usepackage[OT1]{fontenc}
\usepackage{lmodern}
\usepackage[english]{babel}
\usepackage[utf8x]{inputenc}
\usepackage[sc]{mathpazo}
\usepackage{amsmath,amssymb,amsfonts,mathrsfs,amstext}
\usepackage[thmmarks]{ntheorem}
\usepackage{wrapfig}

\usepackage{varioref}

\usepackage{datetime}

\usepackage{mathtools}

\usepackage{array}

\usepackage{listings}
\numberwithin{equation}{section}

\newtheorem{theorem}{Theorem}[section]

\newtheorem{remark}[theorem]{Remark}

\newtheorem{lemma}[theorem]{Lemma}
\newtheorem{proposition}[theorem]{Proposition}

\theoremstyle{nonumberplain}
\theorembodyfont{\normalfont}
\theoremsymbol{\ensuremath{\blacksquare}}
\newtheorem{proof}{Proof}

\renewcommand{\epsilon}{\ensuremath\varepsilon}

\renewcommand{\phi}{\ensuremath{\varphi}}

\usepackage[linkcolor=black,colorlinks=true,citecolor=black,filecolor=black]{hyperref}

\usepackage{algpseudocode}
\usepackage{algorithm}

\usepackage{color}
\usepackage{caption}
\usepackage{subcaption}
\usepackage[noabbrev]{cleveref}
\usepackage{csquotes}
\usepackage{bm}
\usepackage{mathtools}
\usepackage{nicefrac}
\usepackage{booktabs}

\newcommand\mathmidscript[1]{\vcenter{\hbox{$\scriptstyle #1$}}}

\newcommand{\pvint}{\ensuremath{\;\mathmidscript{\raisebox{-0.9ex}{\tiny{p.\,v.\;}}}\!\!\!\!\!\!\!\!\int}}

\DeclareMathAlphabet{\mathpzc}{OT1}{pzc}{m}{it}

\newcommand*\Laplace{\mathop{}\!\mathbin\bigtriangleup}

\newcommand{\NORM}[1]{\left\lVert#1\right\rVert} 

\newcommand{\DEF}{\coloneqq}

\newcommand{\RR}{\mathbb{R}}
\newcommand{\CC}{\mathbb{C}}
\newcommand{\NN}{\mathbb{N}}

\newcommand{\Om}{\Omega}

\newcommand{\del}{\partial}

\newcommand{\MID}{\!\! \mid\!}

\newcommand{\eps}{\epsilon}

\newcommand{\OO}{\mathcal{O}}

\newcommand{\intd}{\mathrm{d}}

\newcommand{\calA}{\mathcal{A}}

\newcommand{\Leu}{\mathrm{L}}

\title{Wave Enhancement through Optimization of Boundary Conditions}
\author{ Habib Ammari\thanks{\footnotesize Department of Mathematics, 
ETH Z\"urich, 
R\"amistrasse 101, CH-8092 Z\"urich, Switzerland (habib.ammari@math.ethz.ch, kthim.imeri@sam.math.ethz.ch).} \and Oscar Bruno\thanks{Computing \& Mathematical Sciences, California Institute of Technology, 1200 E California Blvd, Pasadena, CA 91125, USA (obruno@caltech.edu).} \and Kthim Imeri\footnotemark[1] \and Nilima Nigam \thanks{\footnotesize Department of Mathematics, Simon Fraser University, 8888 University Dr, Burnaby, BC V5A 1S6, Canada (nigam@math.sfu.ca).} }
\date{}

\begin{document}
	\maketitle

\begin{abstract}
	It is well-known that changing boundary conditions for the Laplacian from Dirichlet to Neumann can result in significant changes to the associated eigenmodes, while keeping the eigenvalues close.
We present a new and efficient approach for optimizing the transmission signal between two points in a cavity at a given frequency, by changing boundary conditions.   The proposed approach makes use of 
recent results on the monotonicity of the eigenvalues of the mixed boundary value problem and on the sensitivity of the Green's function to small changes in the boundary conditions. The switching of the boundary condition from Dirichlet to Neumann can be performed through the use of the recently modeled concept of metasurfaces which are comprised of coupled pairs of Helmholtz resonators.  A variety of numerical experiments are presented to show the applicability and the accuracy of the proposed new methodology. 
\end{abstract}

\def\keywords2{\vspace{.5em}{\textbf{  Mathematics Subject Classification
(MSC2000).}~\,\relax}}
\def\endkeywords2{\par}
\keywords2{{35R30, 35C20.}}

\def\keywords{\vspace{.5em}{\textbf{ Keywords.}~\,\relax}}
\def\endkeywords{\par}
\keywords{{Zaremba eigenvalue problem, boundary integral operators, mixed boundary conditions, metasurfaces.}}


\section{Introduction}\label{Ch:Introduction}

This paper develops a new and efficient approach for maximizing  the transmission signal between two points at a chosen frequency through changes to specific eigenmodes of the cavity. These changes are achieved by changing  the  boundary conditions. The eigenmodes and the associated eigenfrequencies of a cavity are sensitively dependant on the geometric properties of the domains, as well as the location of Dirichlet and Neumann boundary conditions.  Many recent works have been devoted to the understanding of the effect of changing the boundary condition on the eigenmodes and the eigenfrequencies \cite{EldarBruno2018, Nigam2014, add2,add3,  HarrellEvans2006, add4, LotoreichikRohleder, add6}.

Through the use of a tunable
reflecting metasurface, the boundary condition can be switched from Dirichlet to
Neumann at some specific resonant frequencies \cite{HRMetasurfaceOnArxiv}. 
In \cite[Part I]{HRMetasurfaceOnArxiv}, the physical mechanism underlying the concept of tunable metasurfaces is modeled both mathematically and numerically. It is shown that an array of coupled pairs of Helmholtz resonators behaves as an equivalent surface with Neumann boundary condition at some specific subwavelength resonant frequencies, where the size of one pair of Helmholtz resonators is much smaller than the wavelengths at
the resonant frequencies. The Green's function of a cavity with mixed (Dirichlet and Neumann) boundary conditions (called also the Zaremba function) is also characterized. In \cite[Part II]{HRMetasurfaceOnArxiv}, a one\--shot optimization algorithm is proposed and used to obtain a good initial guess for the positions around which the boundary conditions should be switched from Dirichlet to Neumann.

In this paper,  we present a new methodology for maximizing the Zarem\-ba function between two points at a chosen frequency through specific eigemodes of the cavity.  The paper is organized as follows. 
In Section \ref{Ch:Prelim}, we first recall some useful results on the eigenvalues of the mixed boundary value problem (called Zaremba eigenvalue problem). Of particular interest is their monotonicity property with respect to the size of the Neumann part proven in \cite{LotoreichikRohleder}. Then we reformulate the eigenvalue problem using boundary integral operators. Based on this nonlinear formulation and the use of the generalized argument principle for the characterization of the characteristic values of finitely meromorphic operator-valued functions of Fredholm-type, we  derive an accurate asymptotic formula of the  changes of eigenfrequencies of a cavity with mixed boundary conditions in terms of the size of the part of the cavity boundary  where the boundary condition is switched from Dirichlet to Neumann. Finally, we recall the asymptotic expansion of the  Zaremba function in terms of the size of the Neumann part. 
The problem of changing a portion of a Dirichlet boundary to Neumann is more delicate than the
converse. If a portion of the boundary is changed from having Neumann conditions to Dirichlet, the reverse consideration than in this paper, then an asymptotic expansion of the eigenvalues is easier to derive \cite{add3, add6}. 
The perturbation theory for the introduction of Neumann boundaries  requires a careful consideration of the asymptotic
behaviour of the Zaremba near the perturbation \cite[Part II]{HRMetasurfaceOnArxiv}. In Section \ref{Ch:DecompZarembaFct}, we derive a spectral decomposition of the Zaremba function. In Section \ref{Ch:Algorithm}, we consider the  problem where we have a source in a bounded domain operating at a given frequency, and we want to
determine, by exploiting the monotonicity property of the eigenvalues of the mixed boundary value problem, which part of the boundary to choose to be reflecting such that an eigenvalue of the mixed boundary value problem gets close enough to the operating frequency. In order to significantly enhance the signal at a given receiving point, both the emitter and the receiver should not belong to the nodal set corresponding to the eigenmode associated with the eigenvalue of the mixed boundary value problem. 

There are two distinct issues: {\it where} to place the Neumann boundary condition, and {\it how long} it should be, to achieve the twin objectives of maximizing gain between a fixed source-receiver pair, and at a frequency close to a desired one.

Our main idea is to first nucleate the Neumann boundary conditions in order to maximize gain of the Zaremba function by making use of an asymptotic expansion of the Zaremba function in terms of the size of the Neumann part. Then the size of the Neumann part is changed in such away that an eigenvalue of the mixed boundary value problem gets close to the operating frequency by using the monotonicity property of the eigenvalues of the mixed eigenvalue problem. The optimization needs the high-accuracy evaluation of certain boundary integral operators, and this is done using techniques from \cite{EldarBruno2018, Nigam2014}.

We present in Section \ref{Ch:NumImplTest}  some numerical experiments to show the applicability and the accuracy of the proposed methodology. 

\section{Preliminaries}\label{Ch:Prelim}

\newcommand{\GDir}{{\Gamma_{\textrm{D}}}}
\newcommand{\GNeu}{{\Gamma_{\textrm{N}}}}
\newcommand{\GDel}{{\Gamma_{\Delta}}}
\makeatletter
\def\moverlay{\mathpalette\mov@rlay}
\def\mov@rlay#1#2{\leavevmode\vtop{%
   \baselineskip\z@skip \lineskiplimit-\maxdimen
   \ialign{\hfil$\m@th#1##$\hfil\cr#2\crcr}}}
\newcommand{\charfusion}[3][\mathord]{
    #1{\ifx#1\mathop\vphantom{#2}\fi
        \mathpalette\mov@rlay{#2\cr#3}
      }
    \ifx#1\mathop\expandafter\displaylimits\fi}
\makeatother

\newcommand{\cupdot}{\charfusion[\mathbin]{\cup}{\cdot}}
\newcommand{\bigcupdot}{\charfusion[\mathop]{\bigcup}{\cdot}}
\newcommand{\OmBar}{\overline{\Om}}
\newcommand{\xS}{x_{\textrm{S}}}
\newcommand{\ZxS}{\mathrm{Z}^k_{\xS}}
\newcommand{\ZxSD}{\mathrm{Z}_{\mathrm{D}, \xS}^k}
\newcommand{\ZxSN}{\mathrm{Z}_{\mathrm{N}, \xS}^k}
\newcommand{\lgdir}{\lambda^{\GDir}}
\newcommand{\lgfull}{\lambda^{\del\Om}}
\newcommand{\lgempt}{\lambda^{\varnothing}}
\newcommand{\Hank}[1]{\mathrm{H}^{(1)}_{#1}}
\newcommand{\Sobo}{{H}}
\newcommand{\Sobominhalf}{{H}^{-\nicefrac{1}{2}}}
\newcommand{\Soboplushalf}{{H}^{\nicefrac{1}{2}}}
\newcommand{\Hminhalftilde}{\widetilde{H}^{-\nicefrac{1}{2}}}
\newcommand{\HminhalftildeNull}{\widetilde{H}^{-\nicefrac{1}{2}}_{\langle 0\rangle}}
\newcommand{\Hplushalfast}{{H}^{\nicefrac{1}{2}}_{\ast}}
\newcommand{\SGDk}{\mathcal{S}_\GDir^k}
\newcommand{\SGNk}{\mathcal{S}_\GNeu^k}
\newcommand{\KGNkstar}{(\mathcal{K}^{k}_\GNeu)^{\ast}}
\newcommand{\KGGkstar}{(\mathcal{K}^{k}_{\del\Om})^{\ast}}
\newcommand{\dSGDk}{\del\mathcal{S}^{k}_{\GDir}}
\newcommand{\Gk}{\Gamma^k}
\newcommand{\calAk}{\mathcal{A}(k)}
\newcommand{\calAe}{\mathcal{A}_\eps}
\newcommand{\calBe}{\mathcal{B}_\eps}
\newcommand{\calAz}{\mathcal{A}_0}
\newcommand{\psiDN}{\begin{bmatrix} \psi\MID_{\GDir}\\ \psi\MID_{\GNeu}\end{bmatrix}}
\newcommand{\calAkz}{\mathcal{A}_0(k)}
\newcommand{\calAke}{\mathcal{A}_\eps(k)}
\newcommand{\kjzer}{k_j^0}
\newcommand{\kjeps}{k_j^\eps}

\subsection{Laplace Eigenvalue with Mixed Boundary Conditions}\label{subsection:LaplEV}

\begin{wrapfigure}{r}{0.42\textwidth}
  \centering
  \includegraphics[width=0.34\textwidth]{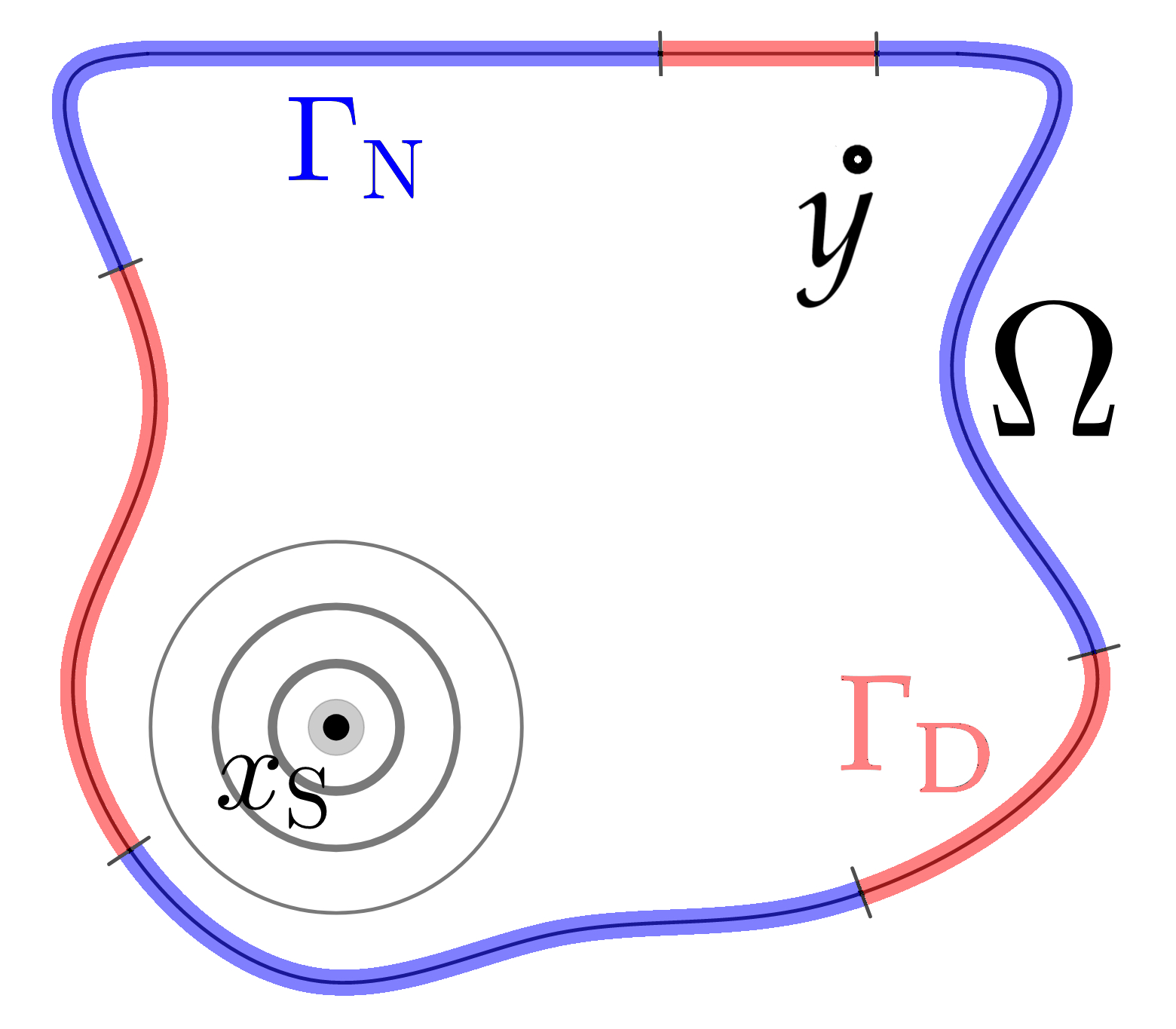}
  \vspace{-10pt}
  \caption{$\GNeu$ is marked in blue and $\GDir$ in red.}\label{fig:PreLimDomain}
\end{wrapfigure}

Let $\Om\subset\RR^2$ be an open, bounded domain with a smooth boundary. We define $\OmBar$ as the topological closure of $\Om$. We decompose the boundary $\del\Om\DEF\OmBar\setminus\Om$ into two parts, $\del\Om=\overline{\GDir \cupdot \GNeu}$, where $\GDir$ and $\GNeu$ are finite unions of open boundary sets. We define $(\GDir,\GNeu)$ to be a partition of $\del\Om$.
Let $\xS\in\Om$ and $k\in(0,\infty)$. The {\it Zaremba function} $\ZxS(x_s,\cdot): \Om\setminus\{\xS\}\rightarrow \RR$ is the Green's function to the Zaremba problem, also known as the fundamental Helmholtz equation with mixed boundary conditions,
\begin{align}\label{pde:ZarembaFunda}
	\left\{ 
	\begin{aligned}
		 \left( \Laplace + k^2  \right) \ZxS(\xS,y)   &= \delta_0(\xS-y) \quad &&\text{for} \; &&y\in\Om\,, \\
		 \ZxS(\xS,y) 								  &= 0 \quad &&\text{for} \; &&y\in\GDir \,,\\
		 \del_{\nu_y} \ZxS(\xS,y) 					  &= 0 \; &&\text{for} \; &&y\in\GNeu. \,
	\end{aligned}
	\right.
\end{align}
Here $\nu_y$ denotes the outer normal at $y\in\del\Om$ and $\del_{\nu_y}$ the normal derivative at $y\in\del\Om$. It is clear that we can write 
$$ 
	Z^k(x_s, \cdot )= \Gamma^k(x_s,\cdot) + \mathrm{R}^k(x_s,\cdot)
$$ 
where $\Gamma^k(x,y)\DEF\frac{i}{4}H_0^1(k\,|x-y| )$ is the fundamental solution of the Helmholtz problem with wavenumber $k$, and $\mathrm{R}^k(x_s,\cdot)$ is a smooth function satisfying the boundary value problem

\begin{align}\label{pde:ZarRemainder}
\left\{ 
	\begin{aligned}
		\left( \Laplace + k^2  \right) \mathrm{R}^k(x_s,y) &=0, \, &&\mbox{in } \,\Omega\,,\\
  		\mathrm{R}^k(x_s,y) 		       				   &=-\Gamma^k(x_s,y), \, &&\mbox{on } \Gamma_D\,, \\   
  		\del_{\nu_y}\mathrm{R}^k(x_s,y) 				   &=-\del_{\nu_y} \Gamma^k(x_s,y) &&\mbox{on }\Gamma_N\,.
	\end{aligned}
	\right.
\end{align}

In Section \ref{Ch:DecompZarembaFct}, we will see that $\ZxS$ exists for all but countably many values of $k$, which are related to the unique solvability of the problem for $\mathrm{R}^k(\xS,\cdot)$. These exceptional values of $k$ are the eigenvalues to the associated Laplace eigenvalue problem with mixed boundary conditions
\begin{align}\label{pde:ZarembaHom}
	\left\{ 
	\begin{aligned}
		 -\Laplace      u		&= k^2 \, u \quad &&\text{in} \; &&\Om\,, \\
		 u      				&= 0 \quad &&\text{on} \; &&\GDir \,,\\
		 \del_{\nu_y}  u		&= 0 \quad &&\text{on} \; &&\GNeu \,.
	\end{aligned}
	\right.
\end{align}
Equation (\ref{pde:ZarembaHom}) has a non-trivial solution $u\in \mathrm{H}^1(\Om)$ for a countable set of real values of $k^2$  \cite[Theorem 4.10]{McLeanEllitpicSystems}, which we refer to as $\{\lgdir_j\}_{j=1}^\infty$, so that $\lgdir_1\leq\lgdir_2\leq\lgdir_3\leq\ldots\,$. We know that $\lgdir_1\geq 0$ and that $\lim_{j\rightarrow\infty}\lgdir_j=\infty$ for all partitions $(\GDir,\GNeu)$ of $\del\Om$.

We denote by $\{\lgfull_j\}_{j\in\NN}$ the pure Dirichlet eigenvalues for $\Omega$,  corresponding to the case $\GDir=\del\Omega$. We let $\{\lgempt_j\}_{j\in\NN}$ denote the Neumann eigenvalues associated to the case $\GNeu=\del\Om$. Then we have 
\begin{align*}
	 0 &< \lgfull_1\,,\quad &&\lgfull_1<\lgfull_2\,, \quad  &&\lgfull_2\leq\lgfull_j\,,\;\forall j\geq3,\\
	 0 &= \lgempt_1\,,\quad &&\lgempt_1<\lgempt_2\,, \quad  &&\lgempt_2\leq\lgempt_j\,,\;\forall j\geq3.
\end{align*}
In \cite{Filonov}, it is shown that $\lgempt_{j+1}<\lgfull_j$, for all $j\in\NN$, for a very general class of domains $\Om$.

\begin{remark}
	Let $\Om$ be the unit circle, we have that $\{\lgfull_j\}_{j=1}^\infty$ is (up to sorting) equal to 
	\begin{align*}
		\{ \alpha^2\in(0,\infty)\MID \exists n\in\NN_0 : \alpha \textrm{ is positive root of } J_n(x) \} \,,
	\end{align*}
	where $J_n$ is a Bessel function of the first kind and order $n$. The eigenvalues corresponding to the roots of $J_0$ appear as simple Dirichlet eigenvalues; all others have multiplicity two. and $\{\lgempt_j\}_{j=2}^\infty$ is (up to sorting) equal to 
	\begin{align*}
		\{ \alpha^2\in(0,\infty)\MID \exists n\in\NN_0 : \alpha \textrm{ is positive root of } J^\prime_n(x) \} \,.
	\end{align*}
	Again, the eigenvalues corresponding to the roots of $J'_0$ appear as simple Neumann eigenvalues; all others have multiplicity two. We refer to \cite{LaplEigenValuesFunctions}.
\end{remark}

Recently, Lotoreichik and Rohleder \cite[Proposition 2.3]{LotoreichikRohleder} showed the following monotonicity statement.
\begin{proposition}\label{prop:lambda<lambda'}
	Let $(\GDir, \GNeu), (\GDir',\GNeu')$ be two partitions of $\del\Om$, such that $\GDir\subset \GDir'$. If $\GDir'\setminus\GDir$ has a non-empty interior then
	\begin{align*}
		\lgdir_j<\lambda^{\Gamma_{D^\prime}}_j\quad \text{ for all } j\in\NN. 
	\end{align*}	 
\end{proposition}

\begin{figure}[h]
  \begin{subfigure}{0.49\textwidth}
    \centering
    \includegraphics[width=0.6\textwidth]{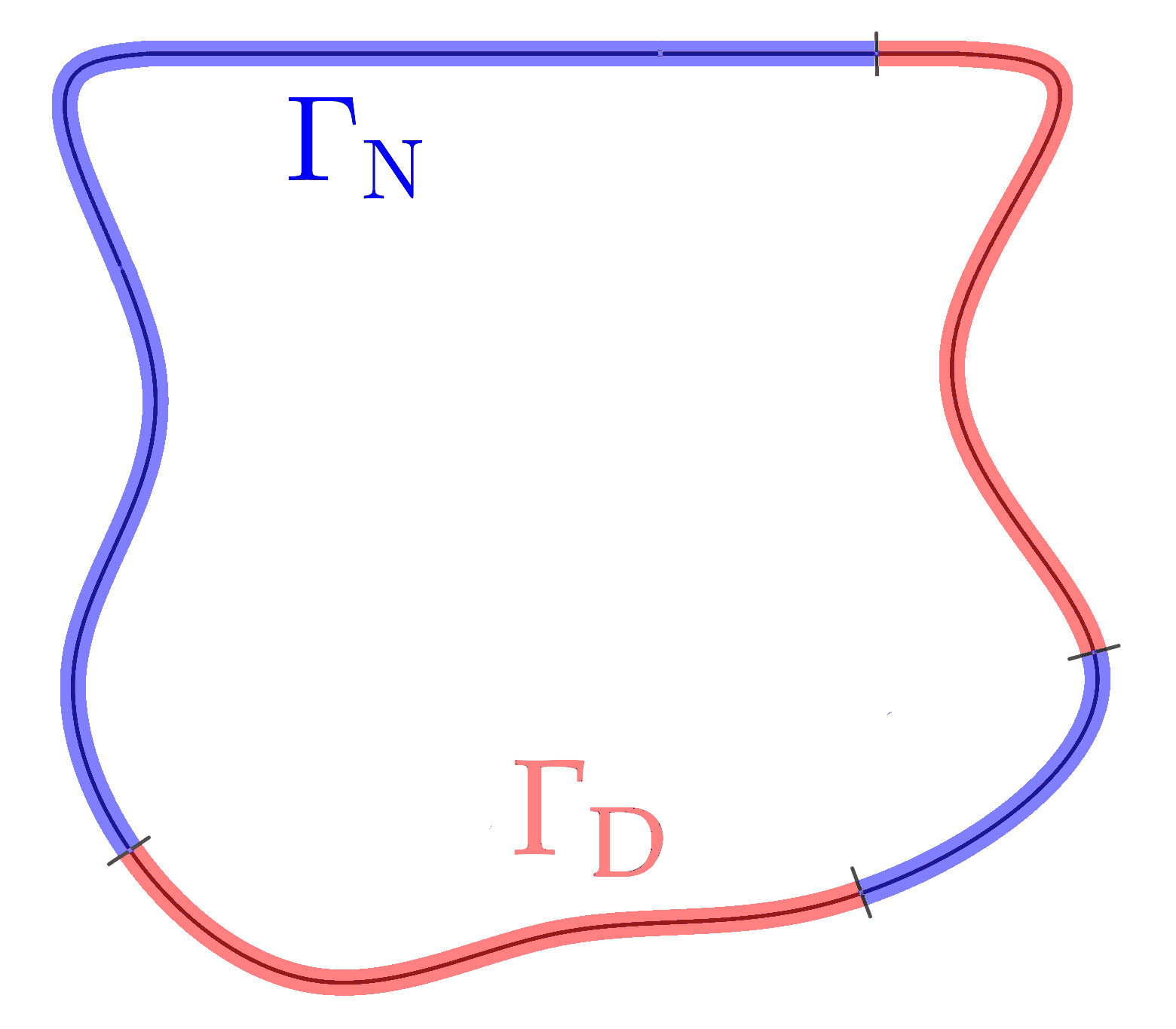}
  \end{subfigure}\hfill 
  \begin{subfigure}{0.49\textwidth} 
    \centering
    \includegraphics[width=0.6\textwidth]{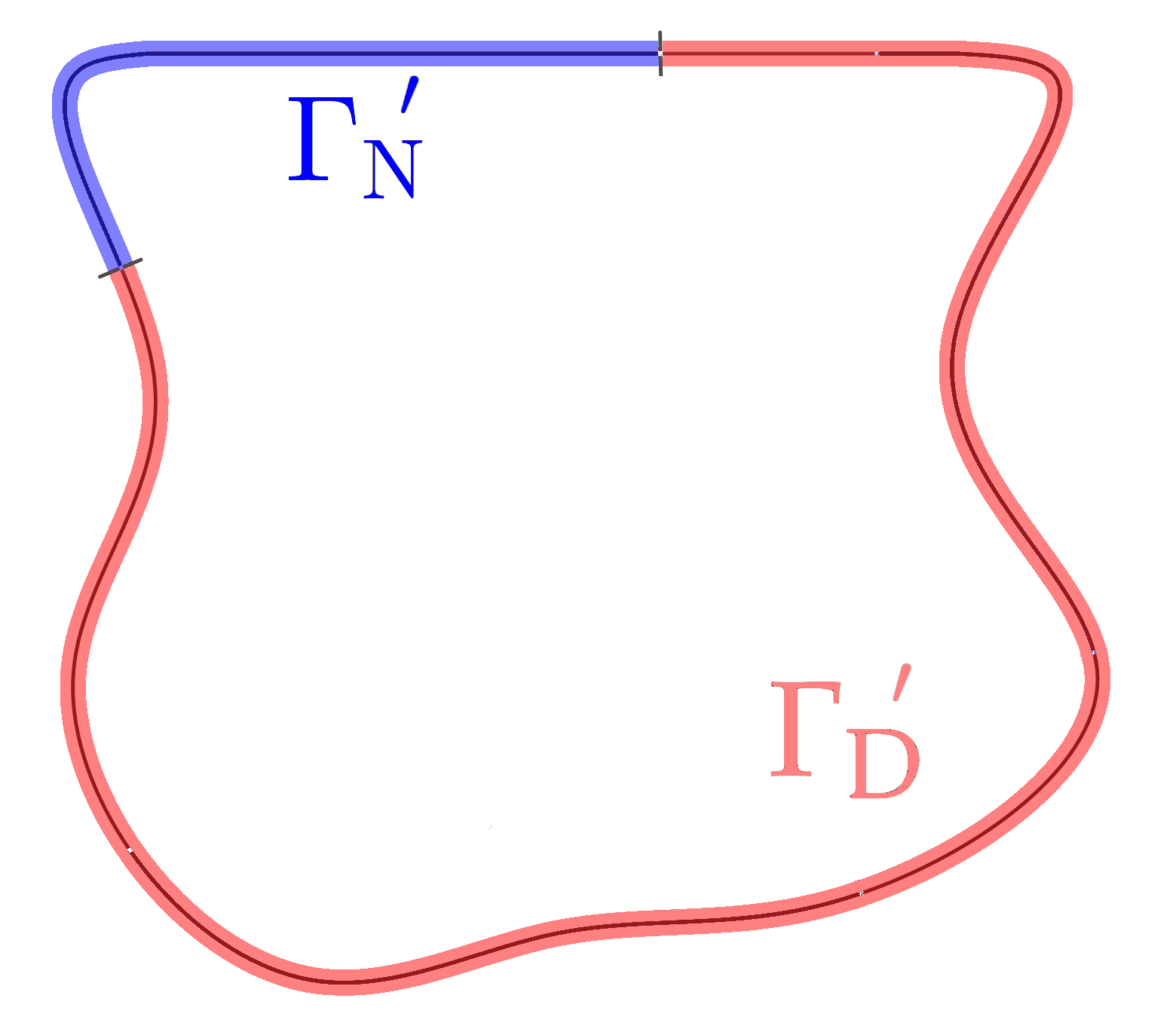}
  \end{subfigure}
  \caption{An illustrative example of the two partitions mentionned in Proposition \ref{prop:lambda<lambda'}. On the left-hand side we have the parition $(\GDel, \GNeu)$ and on the right-hand side $(\GDel',\GNeu')$. They satisfy the condition $\GDel\subset\GDel'$ and that $\GDel'\setminus\GDel$ has a non-empty interior.}
\end{figure}

With Proposition \ref{prop:lambda<lambda'}, we can readily infer that if $\varnothing\neq \GDir$, and $\overline{\GDir}\neq \del\Om$, then
\begin{align*}
	\lgempt_j < \lgdir_j < \lgfull_j\,,\quad \text{ for all } j\in\NN.
\end{align*}

\subsection{Boundary Integral Formulation of the Eigenvalue Problem}\label{subsec:BdryOperatortoEigProb}
The solution $u$ of the eigenvalue(\ref{pde:ZarembaHom}) can be represented by a single layer potential
\begin{align}
	u(x) = \int_{\del\Om}\Gk(x,y)\psi(y)\intd\sigma_y\,,
\end{align}
with surface density $\psi\in L^2(\del\Om)$.

We define then the operators 
$\SGDk: H^{-1/2}(\GDir) \to H^{1/2}(\GDir)$, 
$\SGNk: H^{-1/2}(\GNeu) \to H^{-1/2}(\GDir)$,
$\KGNkstar: H^{-1/2}(\GNeu) \to H^{-1/2}(\GNeu)$ and 
$\dSGDk: H^{-1/2}(\GDir) \to H^{-1/2}(\GNeu)$ by
\begin{align*}
    \SGDk[\psi](x) &\!\DEF\! \int _{\GDir}\Gk(x,y)\psi (y)\intd\sigma_y\,, 
		&& \hspace{-5pt} \SGNk[\psi](x)     \!\DEF\! \int _{\GNeu}\Gk(x,y)\psi (y)\intd\sigma_y\,, \\
    \dSGDk[\psi](x) &\!\DEF\! \!\int _{\GDir}\!\!{\del_{\nu_x}}\Gk(x,y)\psi (y)\intd\sigma_y\,,  
    	&& \hspace{-5pt} \KGNkstar[\psi](x) \!\DEF\! \pvint _{\GNeu}\!\!{\del_{\nu_x}}\Gk(x,y)\psi (y)\intd\sigma_y \,,
\end{align*}
where the 'p.v.' stands for the principle value integral. This   actually is the standard (Lebesgue-) integral for a smooth curved $\GNeu$, since $\del_\nu\Gk$ is a bounded and sufficiently smooth integral operator kernel. From \cite[Chapter 11]{SaranenVainikko2002} we have that $\SGDk$ is a Fredholm operator with index 0, we also readily infer that $\KGNkstar$, $\dSGDk$, and $\SGNk$ are compact operator. 

We then define $\calAk: H^{-1/2}(\GDir)\times H^{-1/2}(\GNeu)\to H^{1/2}(\GDir)\times
H^{-1/2}(\GNeu)$ in terms of these integral operators through
\begin{align}
	\calAk \psiDN \DEF
	\begin{bmatrix}
		\SGDk & \SGNk \\
		\dSGDk & -\frac{1}{2}\mathrm{I}_{L^2_\omega(\GNeu)}+\KGNkstar
	\end{bmatrix}\psiDN \,. \label{equdef:calAk}
\end{align}

We readily see that $\calAk$ is an analytic Fredholm operator of index $0$ in $\CC \setminus i \RR^-$. 

To locate the Zaremba eigenvalues, we have the following statement:
\begin{align}\label{statement:calAk=0}
\mbox{``The real positive characteristics values of the operator-valued function} \nonumber
\\ \mbox{ $k\mapsto\calAk$ are the square roots of the Zaremba eigenvalues''}. 
\end{align}
In \cite[Section 3]{Nigam2014} and \cite{EldarBruno2018}, it is shown that every square root of a Zaremba eigenvalue is a real positive characteristic value of $k\mapsto\calAk$ and every real positive characteristic value of $k\mapsto\calAk$ is the square root of a Zaremba eigenvalue. 

%
We see that $\calAk$ is invertible for $k\in(0,\infty)$ not a square root of a Zaremba eigenvalue.


We remark that the non-real characteristic values of $k\mapsto\calAk$  cannot correspond to eigenvalues to the Laplace equation. This yields the undesirable, but avoidable, difficulty in choosing a neighbourhood $V$ to apply Proposition \ref{prop:keps-ko good approx} in our algorithm, see also Section \ref{Ch:Algorithm}, comment on Line 13. 

The Statement (\ref{statement:calAk=0}) allows for a discretization and thus a numerical approximation of the value $k$. We will use this further on. For these facts, we refer to \cite[Sections 3 and 5]{Nigam2014}.

Let us also consider the regularity of the solution $u$ and the density $\psi$ near a Dirichlet-Neumann junction. The following result can be found in \cite[Theorems 4.2 and 4.3]{EldarBruno2018}.
\begin{proposition}  \label{prop:asympt for u and phi}
Let $\GDir, \GNeu$ be non-empty. Let $k>0$ and $\psi$ satisfy the Statement (\ref{statement:calAk=0}). Let $y_\star\in\overline{\GDir}\cap\overline{\GNeu}$. Then there exists a neighborhood $\mathcal{U}\subset\RR^2$ around $y_\star$ such that for all $y\in\mathcal{U}$ and for all $n\in\NN$
\begin{align*}
	u(y) 
		&= \mathrm{P}^n_{y_\star}(z^{1/2},\overline{z}^{1/2}) + o(z^{n})\,,\\
	\psi\MID_\GDir(y) 
		&= |z|^{-1/2} \mathrm{Q}_{\mathrm{D}, y_\star}^{n}(|z|^{1/2}) + o(|z|^{n-1})\,,  \\
    \psi\MID_\GNeu(y) 
   		&= |z|^{-1/2} \mathrm{Q}_{\mathrm{N}, y_\star}^{n}(|z|^{1/2}) + o(|z|^{n-1})\,,
\end{align*}
where $z\in\CC$ is the complexification of $y-y_\star$, that is $z = \left(y_1-(y_\star)_1\right) + i \left(y_2-(y_\star)_2\right)$ with $i$ being the imaginary unit, and $\overline{z}$ being its conjugate value, and where $P^n_{y_\star}, \mathrm{Q}_{\mathrm{D}, y_\star}^{n}, \mathrm{Q}_{\mathrm{N}, y_\star}^{n}$ are polynomial functions of their respective arguments and of a degree such that none of their terms can be included in their respective error terms.
\end{proposition}

\begin{figure}[h]
    \centering
    \includegraphics[width=0.6\textwidth]{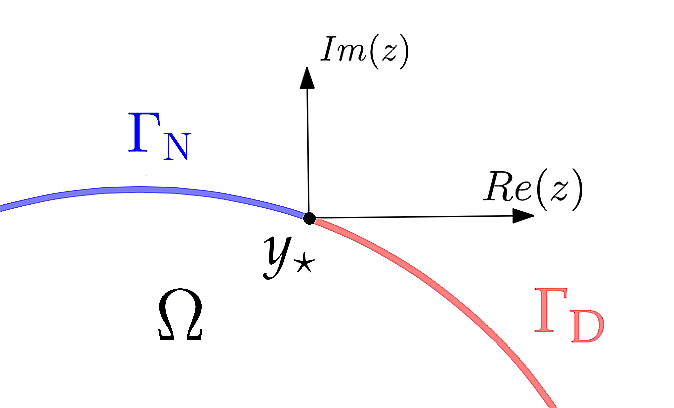}
  \caption{An illustration of the setup used in Proposition \ref{prop:asympt for u and phi}. $z$ is defined as the complexification of a $\RR^2$-vector, with the origin at $y_\star$.}\label{fig:PolynomialDesc}
\end{figure}

\subsection{Approximation of the Zaremba Eigenvalue using the Generalized Argument Principle}\label{subsection:EValDiffwithGenArgPrinc}

In this section we derive asymptotic expressions for the perturbation of the Zaremba eigenvalues, when a small portion of the boundary is changed from Dirichlet to Neumann. 

\begin{figure}[h]
    \centering
    \includegraphics[width=0.35\textwidth]{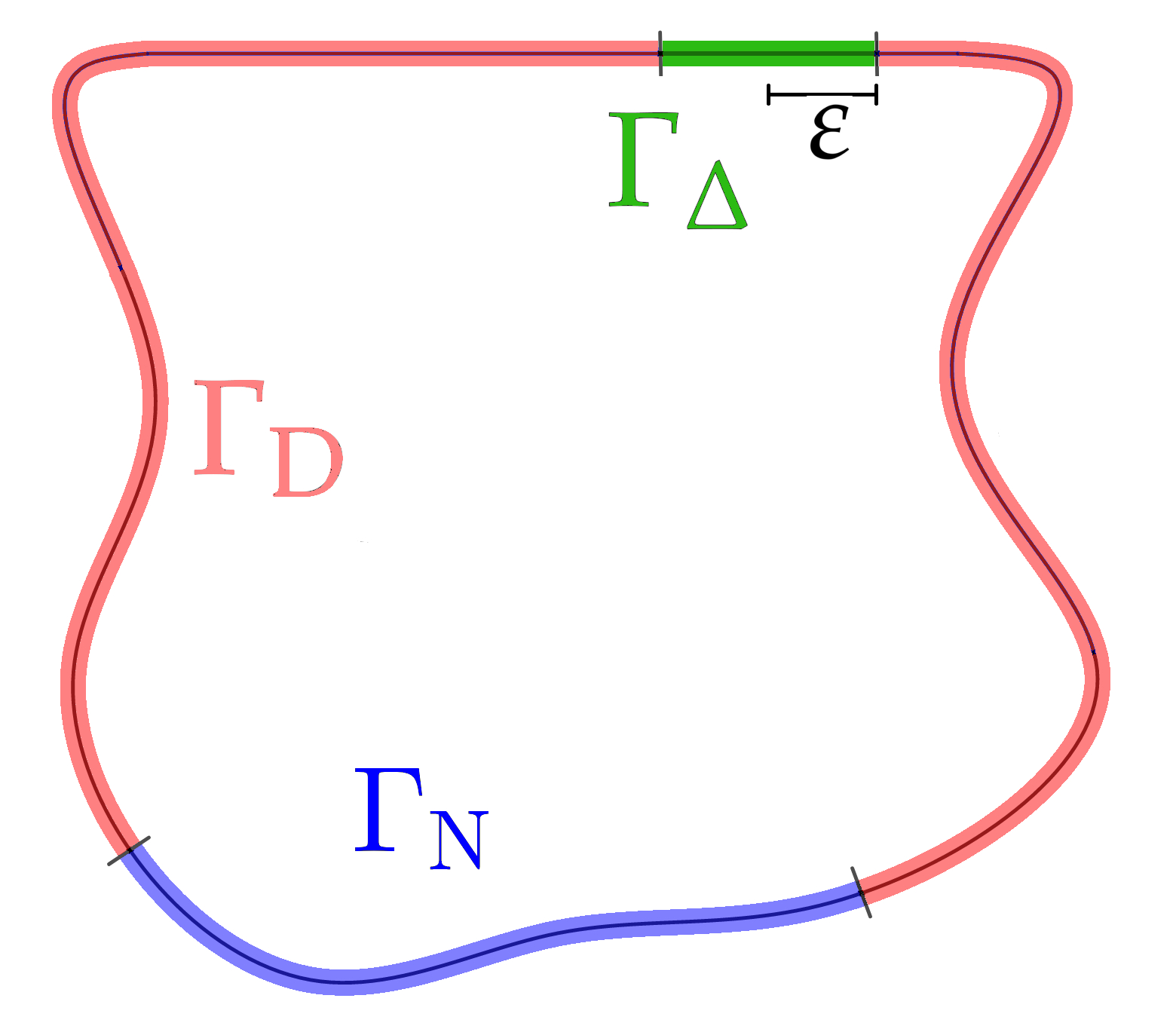}
  \caption{An example for a domain with a Neumann boundary and a Dirichlet boundary and a small straight arc $\GDel$ of length $2\eps$. We associate $\kjzer$ with $\GDel$ being a Dirichlet boundary and $\kjeps$ with $\GDel$ being a Neumann boundary.}
\end{figure}

Let $\GDel\subset \del\Om$ be a boundary interval of length $2\eps$. Let $(\GDir\cupdot\GDel, \GNeu)$ be a partition of $\del\Om$. We  associate the operator $\calAkz$, defined via (\ref{equdef:calAk}), to that partition. This corresponds to $\GDel$ having a Dirichlet boundary condition. Then we define $\calAe(k)$, also by obvious changes in the integrals in (\ref{equdef:calAk}), to be the operator associated to the partition $(\GDir, \GNeu\cupdot\GDel)$. This in turn corresponds to $\GDel$ being a  Neumann part.
For ease of notation, we define $\kjzer\DEF\sqrt{\lambda^{\GDir\cupdot\GDel}_j}$ and $\kjeps\DEF\sqrt{\lambda^{\GDir}_j}$ for all $j\in\NN$, and call those characteristic values to their respective operators. From \cite[Lemma 3.8]{LPTSA} we then have the following Lemma:

\begin{lemma}\label{lemma:GenArgPri Exact}
	Let $\kjzer$ be a simple characteristic value. Let $V\subset\CC$ be a neighbourhood of $\kjzer$, such that $\kjeps\in V$. Assume further that no other square root of Zaremba eigenvalue to the partition $(\GDir, \GNeu\cupdot\GDel)$ of $\del\Om$ is in  $\overline{V}$. Then $\kjeps$ is given by the contour integral
	\begin{align*}
		\kjeps-\kjzer 
			= \frac{1}{2\pi i}\,\mathrm{tr}
				\int_{\del V} (\omega-\kjzer)\calAe(\omega)^{-1} \del_\omega \calAe(\omega)\,\intd \omega\,.
	\end{align*}
\end{lemma}Here $\del_\omega$ denotes the variation of the operator in the wavenumber parameter $\omega$. This expression is exact. Unfortunately, its use in a practical algorithm is limited, since it would entail inverting the operator $\calAe(\omega)$ for each $\epsilon$ used in an optimization. It is useful, therefore, to locate an expression in which this inverse is approximated by $\calAz(\omega)$ instead.

From \cite[Theorem 3.12]{LPTSA} we get the approximation
\begin{align}\label{equ:calAk bad approx}
	\kjeps-\kjzer 
		\approx \frac{-1}{2\pi i}\,\mathrm{tr}
			\int_{\del V} \calAz(\omega)^{-1} (\calAe(\omega)-\calAz(\omega))\,\intd \omega\,,
\end{align}
where we expect the error to be in $o\left( \frac{1}{|\log(\eps)|} \right)$. 
We can, in fact, obtain a faster and even more accurate approximation, which we describe in the following proposition.


\begin{proposition}\label{prop:keps-ko good approx}
	Let $\kjzer$ be the $jth$ (sorted) characteristic value of $\calAkz$ corresponding to the decomposition $\Gamma_{\textrm{D}}, \Gamma_{\textrm{N}}, $ and assume it is simple  Then one can find a $\eps>0$ and a neighbourhood $V\subset \CC$ containing $\kjzer$ so that
	\begin{itemize}
		\item the $jth$ characteristic value $\kjeps$ of the  operator $\calAke$ (obtained by changing $\Gamma_\Delta$ to a Neumann boundary condition) is contained $\in V$;
		\item  no other square root of the Laplace eigenvalues to the partition $(\GDir, \GNeu\cupdot\GDel)$ of $\del\Om$ are in  $\overline{V}$.
		\item The characteristic value of the perturbed operator $\kjeps$ is given by 
	\begin{align*}
		\kjeps-\kjzer 
			= \frac{-1}{2\pi i}\,\mathrm{tr}
				\int_{\del V} (\mathrm{I}+(\omega-\kjzer)\calAe(\kjzer)^{-1}\del_\omega\calAe(\kjzer))^{-1} \,\intd \omega \\
				\times \bigg[ 1 + \OO(
				\mathrm{tr}
				\int_{\del V} (\mathrm{I}+(\omega-\kjzer)\calAe(\kjzer)^{-1}\del_\omega\calAe(\kjzer))^{-1} \,\intd \omega )\bigg].
	\end{align*}\end{itemize}
	Here $\mathrm{I}$ is the identity operator.
\end{proposition}

\begin{proof}
%
	 We first observe from Proposition \ref{prop:lambda<lambda'} together with the fact that $\calAe(k)$ is a Fredholm analytic operator of index $0$ in $\CC \setminus i \RR^-$, we can see that $\kjeps\nearrow\kjzer$ for $\eps\searrow 0^+$. We now examine the {\it perturbed} operator $\calAe$. Its characteristic value is $\kjeps$. Provided $\kjzer$ is sufficiently close to $\kjeps$, we have the following Taylor expansion: 
	\begin{align}\label{equ:Taylor calAke}
		\calAe(\omega)=\calAe(\kjzer)+(\omega-\kjzer)\del_\omega \calAe(\kjzer)+\calBe(\omega)\,,
	\end{align}
	where $\calBe(\omega)=\OO((\omega-\kjzer)^2)$. This expansion holds only in a neighborhood $V^0_\eps$ of $\kjzer$, and so  $\eps$ must be small enough such that $\kjzer\in V^0_\eps$. 
	
	Then consider that we have in the operator-norm 
	\begin{align*}
		\NORM{\left(\calAe(\kjzer)+(\omega-\kjzer)\del_\omega \calAe(\kjzer)\right)^{-1}\calBe(\omega)}<1\,,
	\end{align*}
	for $\omega\in V_\eps\subset V^0_\eps$ close enough to both $\kjeps$ and $\kjzer$, because then the Taylor remainder $\calBe(\omega)=\OO((\omega-\kjzer)^2)$. If $\eps$ is small enough, then $\kjeps\in V_\eps$. Then by the Generalization of Rouché’s Theorem \cite[Theorem 1.15]{LPTSA} we have that since $\calAe(\kjzer)+(\omega-\kjzer)\del_\omega \calAe(\kjzer)$ and $\calAe(\omega)$ are close in operator norm, they both have the same number of characteristic values in $V_\eps$. Thus $\calAe(\kjzer)+(\omega-\kjzer)\del_\omega \calAe(\kjzer)$ has a simple characteristic value $k^{\sharp}_j$ in $V_\eps$.
	Now we can use Lemma \ref{lemma:GenArgPri Exact}, but replacing $\calAz(\omega)$ by $\left(\calAe(\kjzer)+(\omega-\kjzer)\del_\omega \calAe(\kjzer)\right)$:
	
	to get
	\begin{align*}
		k^{\sharp}_j-\kjzer
			&= \frac{1}{2\pi i}\,\mathrm{tr}
				\int_{\del V_\eps} 
					(\omega-\kjzer)
					\left(\calAe(\kjzer)+(\omega-\kjzer)\del_\omega \calAe(\kjzer)\right)^{-1} 
					\\ &\times \del_\omega \left(\calAe(\kjzer)+(\omega-\kjzer)\del_\omega \calAe(\kjzer)\right)\,
				\intd \omega\\
			&= \frac{1}{2\pi i}\,\mathrm{tr}
				\int_{\del V_\eps} 
					(\omega-\kjzer)
					\left(\calAe(\kjzer)+(\omega-\kjzer)\del_\omega \calAe(\kjzer)\right)^{-1} 
					\del_\omega \calAe(\kjzer)
				\,\intd \omega,
		\end{align*} 
		and hence, 
		\begin{align*}
		k^{\sharp}_j-\kjzer		
			&= \frac{1}{2\pi i}\,\mathrm{tr}
				\int_{\del V_\eps} 
					\left(\calAe(\kjzer)+(\omega-\kjzer)\del_\omega \calAe(\kjzer)\right)^{-1} 
					\\ & \times \left(\calAe(\kjzer)+(\omega-\kjzer)\del_\omega \calAe(\kjzer) -\calAe(\kjzer)\right)
				\,\intd \omega\\
			&= \frac{1}{2\pi i}\,\mathrm{tr} \left(
				\int_{\del V_\eps} 
					\mathrm{I}
				\,\intd \omega
				\!-\!
				\int_{\del V_\eps} 
					\!\left(\calAe(\kjzer)\!+\!(\omega\!-\!\kjzer)\del_\omega \calAe(\kjzer)\right)^{-1} 
					\!\!\calAe(\kjzer)
				\,\intd \omega\right)\\
			&= -\frac{1}{2\pi i}\,\mathrm{tr} 
				\int_{\del V_\eps} 
					\left(\calAe(\kjzer)+(\omega-\kjzer)\del_\omega \calAe(\kjzer)\right)^{-1} 
					\calAe(\kjzer)
				\,\intd \omega\,.
	\end{align*}
Moreover, by a standard perturbation argument \cite[Section 5.2.4]{LPTSA}, we have at the leading-order term
$$
\kjeps-	k^{\sharp}_j = - {(\calBe(k^{\sharp}_j) \psi^{\sharp}_j, \psi^{\sharp}_j)}, 
	$$
where $\psi^{\sharp}_j$ is the root function associated with the characteristic value	$k^{\sharp}_j$ evaluated at $k^{\sharp}_j$. Thus, 
\begin{align*}
		\kjeps-\kjzer 
			= (k^{\sharp}_j - \kjzer) (1 + \OO( k^{\sharp}_j - \kjzer))\,,
	\end{align*}
and therefore, Proposition \ref{prop:keps-ko good approx} holds. 
\end{proof}
We remark on the significance of this result from the point of view of computation, and which makes it a key ingredient in our algorithm. If one seeks a high-accuracy approximation of the characteristic value $\kjeps$ of $\calAke$, and one already has a good approximation of $\kjzer$, the approximation in Proposition 2.5 allows us to proceed by assembling only one matrix, that corresponding to $\mathcal{A}_\epsilon(k_j^0)$. The contour integrals can be effectively computed using the trapezoidal rule, making this an inexpensive but very accurate approximation of $\kjeps$.

\subsection{Approximation of the Zaremba Function}\label{subsection:DiffApproxZarembaFct}
Let $\Om\DEF\{ z\in\CC \mid |z|<1 \}$, and let $\GDel\subset \del\Om$ be a boundary interval of length $2\eps$ with center $y_\star\in\GDel$. Let $(\del\Om, \varnothing)$ be the partition of $\del\Om$, and with it we associate the Zaremba function $\mathrm{Z}_{\mathrm{D}}^k(\xS,\cdot)$, for $\xS\in\Om$, defined via  (\ref{pde:ZarembaFunda}). This corresponds to $\GDel$ having a Dirichlet boundary condition. Then we define $\mathrm{Z}_{\mathrm{N}}^k(\xS,\cdot)\in\mathrm{L}^2(\Om)$, $\xS\in\Om$, also defined via (\ref{pde:ZarembaFunda}), to be the Zaremba function associated to the partition $(\del\Om\setminus\overline{\GDel}, \GDel)$. This in turn corresponds to $\GDel$ having a Neumann boundary condition. We then have the following lemma.

\begin{lemma}\label{lemma:ZN-ZD}
	Let $\Om, y_\star, \GDel, \mathrm{Z}_{\mathrm{D}}^k(\xS,\cdot)$ and $\mathrm{Z}_{\mathrm{N}}^k(\xS,\cdot)$ be defined as described above. Let $\eps>0$ be small enough. Let $k>0$, such that $k^2\neq \lambda_j^{\del\Om\setminus\GDel}$, and  $k^2\neq \lambda_j^{\del\Om}$, for all $j\in\NN$. Then for all $z\in\Om$,
	\begin{align*}
		\mathrm{Z}_{\mathrm{N}}^k(\xS,z)
			= \mathrm{Z}_{\mathrm{D}}^k(\xS,z)
			-\eps^2\frac{\pi}{2} \del_{\nu_{y_\star}} \mathrm{Z}_{\mathrm{D}}^k(z,y_\star)\del_{\nu_{y_\star}} \mathrm{Z}_{\mathrm{D}}^k(\xS,y_\star)
			+\OO\left(\frac{\eps^2}{|\log(\eps/2)|^2}\right)\,.
	\end{align*}
\end{lemma}
Lemma \ref{lemma:ZN-ZD} follows readily from combining the results in \cite[Theorem 5.4]{HRMetasurfaceOnArxiv} and \cite[Equation (6.24)]{HRMetasurfaceOnArxiv} 

Numerical experiments confirm that $|\ZxSN(y)-\ZxSD(y)|$ is of order of $\epsilon^2$, as long as $y$ is far enough away from the boundary.

\section{Spectral Decomposition of the Zaremba Function}\label{Ch:DecompZarembaFct}

\newcommand{\HrmG}{\mathrm{H}_{0, \GDir}^1}
\newcommand{\rmr}{\mathrm{r}}
Let us again consider the more general setup at the beginning of Section \ref{Ch:Prelim}, that is let $(\GDir,\GNeu)$ be a partition of $\del\Om$, let $\{ \lgdir_j \}_{j=1}^{\infty}$ be the Zaremba eigenvalues and let $\{ u_j \}_{j=1}^\infty$ be an $L^2$-orthonormal basis of associated eigenfunctions. Then we have the following statement about the Zaremba function $\ZxS$, $\xS\in\Om$, defined by (\ref{pde:ZarembaFunda}).

\begin{theorem}\label{thm:ZarembaDecomp}
	For all $y\in\Om$, $y\neq x_S$ and for all $k>0$ which are not in the spectrum, ie,  $k^2\neq \lgdir_j$ of the Zaremba eigenvalue problem, the Zaremba function $\ZxS$, given by  (\ref{pde:ZarembaFunda}), exists and is in $ \Leu^2_{\text{loc}}(\Om)$. Furthermore, we can write it as
	\begin{align*}
		\ZxS(y)
			= \sum_{j=1}^\infty \frac{u_j(x_S)\,u_j(y)}{k^2-\lgdir_j} \,.
	\end{align*}
\end{theorem}

Next, we will consider the proof of Theorem \ref{thm:ZarembaDecomp}. To this end, we define 
\begin{align*}
	\HrmG(\Om) \DEF \{ v\in\mathrm{H}^1(\Om) \mid v\MID_\GDir = 0 \}\,.
\end{align*}
Consider that the solution to the Laplace eigenvalue-equation $u$ is element of $\HrmG(\Om)$.
$$
	\mathrm{dom}(-\Laplace)\DEF \{ w\in\HrmG(\Om)\mid \Laplace w \in \Leu^2(\Om), \del_\nu w\MID_\GNeu = 0\}\,.
$$
The operator $-\Laplace$ is selfadjoint in $\Leu^2(\Om)$, which we readily see using Green's identity, and it has thus a discrete spectrum. Moreover, $-\Laplace$ corresponds to the sesquilinear form $\langle v_1, v_2 \rangle \mapsto (\nabla v_1, \nabla v_2)_{\Leu^2(\Om)}$ with domain $\HrmG$, since $(-\Laplace w_1, w_2)_{\Leu^2(\Om)}= (\nabla w_1, \nabla w_2)_{\Leu^2(\Om)}$ for all $w_1, w_2\in \mathrm{dom}(-\Laplace)$, see \cite{BS87,Kato,S12} for more details on semi-bounded self-adjoint operators and corresponding quadratic forms. And the form $\langle\cdot\,,\cdot\rangle$ is closed, non-negative and symmetric.
This allows us to use the min-max principle. Thus we can write for all $j\in\NN$,
\begin{align}\label{equ:min-max}
	\lgdir_j = \min_{\substack{L\subset\HrmG(\Om)\\\mathrm{dim} \,L\,=\, j}}\max_{\;\;v\in L\setminus\{0\}}\frac{\NORM{\nabla v}_{\Leu^2(\Om)}^2}{\NORM{v}_{\Leu^2(\Om)}^2}\,.
\end{align}

This leads us to the following lemma.

\begin{lemma}\label{lemma:f-sum c u goes to 0}
	For all $f\in \mathrm{dom}(-\Laplace)$, we have that
	\begin{align}
		\NORM{f-\sum_{j=1}^N c_j\,u_j}_{\Leu^2(\Om)}^2
			= \int_\Om \Big|f - \sum_{j=1}^N c_j\,u_j \Big|^2\intd x \xrightarrow{N\rightarrow \infty}0\,,
	\end{align}
	where $c_j\DEF(f\,, u_j)_{\Leu^2(\Om)}$, that is the linear subset spanned by eigenfunctions of the Laplace eigenvalue-equation with mixed boundary conditions (\ref{pde:ZarembaHom}) is dense in $\mathrm{dom}(-\Laplace)$.
\end{lemma}

\begin{proof}
	Let $\rmr_N\DEF f-\sum_{j=1}^N c_j\,u_j$. Then for all $i=1,\ldots,N$, we have that
	\begin{align*}
		(\rmr_N\,, u_i)_{\Leu^2(\Om)}
			&= \left( f-\sum_{j=1}^N c_j\,u_j\,, u_i \right)_{\Leu^2}
			= (f\,, u_i)_{\Leu^2}- c_i\,(u_i\,,u_j)_{\Leu^2}
			= 0\,,\\
		(\nabla\rmr_N\,, \nabla u_i)_{\Leu^2(\Om)}
			&= (\nabla f\,, \nabla u_i)_{\Leu^2}- \sum_{j=1}^N c_j\,(\nabla u_j\,,\nabla u_i)_{\Leu^2}\\
			&= {\lgdir_i}(f\,, u_i)_{\Leu^2} - {\lgdir_i} c_j\,(u_i\,, u_i)_{\Leu^2}
			=0\,,
	\end{align*}
	where we used Green's identity and the fact that $f, u_j\in \mathrm{dom}(-\Laplace)$. Next, we want to show that 
	\begin{align}\label{equ:lgdirN smaller than nablarN/rN}
		\lgdir_N\leq\frac{\NORM{\nabla \rmr_N}^2_{\Leu^2(\Om)}}{\NORM{\rmr_N}^2_{\Leu^2(\Om)}}.
	\end{align} 
	To this end, consider the min-max principle (\ref{equ:min-max}), it tells us that
	\begin{align*}
		\lgdir_j 
		&\leq \max_{v\in \mathrm{span}\{ u_1,\ldots,u_{N-1}, \rmr_N \}}
					\frac{\NORM{\nabla v}_{\Leu^2(\Om)}^2}{\NORM{v}_{\Leu^2(\Om)}^2}\,\\
		&=\max_{a_1,\ldots, a_N \in \RR}
					\frac{\NORM{\nabla(a_N\,\rmr_N + a_1\,v_1+\ldots +a_{n-1}\,v_{n-1})}^2}
					{\NORM{a_N\,\rmr_N + a_1\,v_1+\ldots a_{n-1}\,v_{n-1}}^2}\,\\
		&=\max_{a_1,\ldots, a_N \in \RR}
					\frac{a_N^2\NORM{\nabla\rmr_N}^2 + a_1^2\NORM{\nabla v_1}^2+\ldots + a_{n-1}^2\NORM{\nabla v_{n-1}}^2}
					{a_N^2\NORM{\rmr_N}^2 + a_1^2\NORM{v_1}^2+\ldots + a_{n-1}^2\NORM{v_{n-1}}^2}\,\\
		&=\max_{a_1,\ldots, a_N \in \RR}
					\frac{a_N^2\NORM{\nabla\rmr_N}^2 + \lgdir_1 a_1^2+\ldots + \lgdir_{n-1} a_{n-1}^2}
					{a_N^2\NORM{\rmr_N}^2 + a_1^2+\ldots + a_{n-1}^2}\,\\
		&\leq\max_{a_1,\ldots, a_N \in \RR}
					\frac{a_N^2\NORM{\nabla\rmr_N}^2 + \lgdir_{n-1}( a_1^2+\ldots + a_{n-1}^2)}
					{a_N^2\NORM{\rmr_N}^2 + a_1^2+\ldots + a_{n-1}^2}\,.
	\end{align*}
	Thus, we can infer $\lgdir_N\leq\frac{\NORM{\nabla \rmr_N}^2}{\NORM{\rmr_N}^2}$ from $\lgdir_{n-1}\leq\frac{\NORM{\nabla \rmr_N}^2}{\NORM{\rmr_N}^2}$, which in turn is given by an induction argument, whose induction basis follows trivially from the min-max principle (\ref{equ:min-max}).
	Using the definition of $c_j$, we have that
	\begin{align*}
		\NORM{\nabla \rmr_N}^2_{\Leu^2(\Om)} 
			&=		\NORM{\nabla f}^2_{\Leu^2(\Om)} 
					-2 \sum_{j=1}^N c_j\,\lgdir_j (f,u_j)_{\Leu^2(\Om)} 
					+ \sum_{j=1}^N c_j^2\,\lgdir_j \NORM{u_j}^2_{\Leu^2(\Om)} \\
			&=		\NORM{\nabla f}^2_{\Leu^2(\Om)} 
					-\sum_{j=1}^N \lgdir_j (f,u_j)^2_{\Leu^2(\Om)} \\
			&\leq	\NORM{\nabla f}^2_{\Leu^2(\Om)}\,.
	\end{align*}
	Thus, using (\ref{equ:lgdirN smaller than nablarN/rN}), we have that 
	\begin{align}
		\NORM{\rmr_N}^2_{\Leu^2}\leq\frac{\NORM{\nabla f}^2_{\Leu^2}}{\lgdir_N}\,.
	\end{align}
Since $\NORM{\nabla f}^2_{\Leu^2} = (f,-\Laplace f)_{\Leu^2}\leq \NORM{f}_{\Leu^2}\NORM{\Laplace f}_{\Leu^2}<\infty$, $\NORM{\nabla f}^2_{\Leu^2}$ is bounded. Using the fact that $\lgdir_N \xrightarrow{N\rightarrow\infty}\infty$, we have that $\NORM{\rmr_N}_{\Leu^2}^2\xrightarrow{N\rightarrow\infty}0$. This completes the proof of Lemma \ref{lemma:f-sum c u goes to 0}.
\end{proof}

\begin{proof}[Theorem \ref{thm:ZarembaDecomp}]
	To show the existence of the Zaremba function $\ZxS$, we write $\ZxS(y)$, for all $y\in\Om\,, y\neq \xS$ as
	\begin{align}
		\ZxS(y) = \Gk(\xS, y) + \mathrm{R}^k(\xS, y)\,,
	\end{align}	 
	where $\Gk$ is the fundamental solution to the Helmholtz equation, and $\mathrm{R}^k$ satisfies
	\begin{align}\label{pde:ZarembaRemainder}
	\left\{ 
		\begin{aligned}
		 	(\Laplace+k^2) \mathrm{R}^k(\xS\,, y)	&= 0 \quad &&\text{in} \; &&\Om\,, \\
		 	\mathrm{R}^k(\xS\,, y)      			&= - \Gk(\xS, y) \quad &&\text{on} \; &&\GDir \,,\\
		 	\del_{\nu_y}  \mathrm{R}^k(\xS\,, y)	&= - \del_{\nu_y}\Gk(\xS, y) \quad &&\text{on} \; &&\GNeu \,.
		\end{aligned}
	\right.
	\end{align}
	The solution to (\ref{pde:ZarembaRemainder}) does exist, for those values of $k$ specified in the theorem, and it is in $\mathrm{H}^1(\Om)$, see \cite[Theorem 4.10]{McLeanEllitpicSystems}. Using that $\Gk(\xS,\cdot)\in L^2(\Om)$, we have that $\ZxS(y)\in L^2(\Om)$. Thus from Lemma \ref{lemma:f-sum c u goes to 0} and the density of $\mathrm{dom}(-\Laplace)$ in $L^2(\Omega)$, we have that for all $y\in\Om$, $y\neq \xS$,
	\begin{align*}
		\ZxS(y)=\sum_{j=1}^\infty a_j u_j(y)\,,
	\end{align*}
	for some $a_j\in\RR$, depending on $\xS$. Let us give an expression for the $a_j$. Using Green's identity, we have that
	\begin{align*}
	u_i(\xS)
		&= \int_\Om (\Laplace+k^2) \ZxS(y)\,u_i(y)\intd y
			= \int_\Om \ZxS(y)\,(\Laplace+k^2) u_i(y)\intd y \\
		&= (k^2-\lgdir_i)\int_\Om \ZxS(y)\,u_i(y)\intd y
			= (k^2-\lgdir_i)\int_\Om \;\sum_{j=1}^\infty\; a_j u_j(y)\,u_i(y)\,\intd y\\
		&= (k^2-\lgdir_i)\sum_{j=1}^\infty a_j \;\delta_0(i-j)
			= (k^2-\lgdir_i)a_i\,,
	\end{align*}
	where we used Fubini's theorem to interchange summation and integration.
	With that we infer that for all $i\in\NN$,
	\begin{align*}
		a_i = \frac{u_i(\xS)}{k^2-\lgdir_i}\,,
	\end{align*}
	and this concludes the proof.
\end{proof}

\makeatletter
\newenvironment{breakablealgorithm}
  {
   \begin{center}
     \refstepcounter{algorithm}
     \hrule height.8pt depth0pt \kern2pt
     \renewcommand{\caption}[2][\relax]{
       {\raggedright\textbf{\ALG@name~\thealgorithm} ##2\par}%
       \ifx\relax##1\relax 
         \addcontentsline{loa}{algorithm}{\protect\numberline{\thealgorithm}##2}%
       \else 
         \addcontentsline{loa}{algorithm}{\protect\numberline{\thealgorithm}##1}%
       \fi
       \kern2pt\hrule\kern2pt
     }
  }{
     \kern2pt\hrule\relax
   \end{center}
  }
\makeatother

\section{The Algorithm}\label{Ch:Algorithm}
We next present our main algorithm for wave enhancement.
We begin with a domain $\Omega$, the source point $\xS$ and the receiver point $y$, both in $\Omega$, and a predetermined target value ${k_\star}$ corresponding to a desired transmission frequency. 

First, we determine the next higher Dirichlet eigenvalue to $k_\star^2$, which is done using a discretized version of the operator $\calA(k)$ given in Section \ref{Ch:Prelim}. The discretization follows the procedure developed in \cite{Nigam2014}.

Second, we determine a location $y_\star$ on the boundary $\del\Om$, which yields a higher absolute value of $|\ZxS(\xS,y)|$, when we insert a small enough Neumann boundary at that location. Finding the location is established using Lemma \ref{lemma:ZN-ZD}, that is we find the local maxima or minima of 
$$\del_{\nu_{y_\star}} \!\mathrm{Z}_{\mathrm{D}}^{k_\star}(\xS,y_\star)\cdot\del_{\nu_{y_\star}} \!\mathrm{Z}_{\mathrm{D}}^{k_\star}(y,y_\star).$$ 
The computation of the Zaremba function is done by solving the problem \ref{pde:ZarRemainder} using the procedure described in \cite{EldarBruno2018}, also uses the operator $\calA(k)$.

Third, we successively increase the Neumann boundary until the characteristic value hits the target characteristic value. The computation of the new characteristic value after a small increase of the Neumann boundary is achieved using Proposition \ref{prop:keps-ko good approx}. It might be that we need to increase the boundary initially by a large amount, and the resulting characteristic value has to be computed with the time-expensive procedure described in \cite{Nigam2014}.
%

A more detailed explanation is given in the comments after Algorithm \ref{Algorithm}.
We note here that changing a boundary part from the Dirichlet boundary condition to the Neumann one, the associated Laplace eigenvalue $\lambda_j^{\GDir}$ decreases, according to Proposition \ref{prop:lambda<lambda'}, and thus the characteristic value $\sqrt{\lambda_j^{\GDir}}$ decreases as well. Moreover, $\lambda_j^{\GDir}$ is between the Neumann and the Dirichlet eigenvalue, that is $\lambda_j^{\varnothing}\leq\lambda_j^{\GDir}\leq\lambda_j^{\del\Om}$. Increasing boundary length enough, we eventually hit the target characteristic value $k_\star$, because $\cup_{j=1}^\infty \Big(\lambda_j^{\varnothing}, \lambda_j^{\del\Om}\Big) = (0, \infty)$, since $\lambda_{j+1}^{\varnothing} < \lambda_j^{\del\Om}$, proved in \cite{Filonov}.

\begin{breakablealgorithm}
\caption{Finding an intensity maximizing partition of the boundary}\label{Algorithm}
\hspace*{\algorithmicindent} \textbf{Input:} $\eps>0$, $\xS\in\Om$, $y\in\Om$, $y\neq \xS$, $k_\star>0$, $C_{\mathrm{tol}}>0$. \\
\textbf{Require:} $\eps$ is small enough, $C_{\mathrm{tol}}$ is big enough. 
\begin{algorithmic}[1] 
\State Let $\GDir\DEF\del\Om$, $\GNeu\DEF\varnothing$.
\State Find the next higher square root of the Dirichlet eigenvalue $k$ to $k_\star$.
\State Compute the value $\mathrm{Z}^k_{(\GDir, \GNeu)}(\xS,y)$ and the normal derivative of the Zaremba functions $\del_{\nu_{\cdot}}\mathrm{Z}^k_{(\GDir, \GNeu)}(\xS,\cdot), \del_{\nu_{\cdot}}\mathrm{Z}^k_{(\GDir, \GNeu)}(y,\cdot)$ associated to the partition $(\GDir, \GNeu)$ at the boundary.

\If{$\mathrm{Z}^k_{(\GDir, \GNeu)}(\xS,y)\geq 0$}
	\State Let $\mathsf{S}$ be the location of a global minima of the function $\del\Om\ni z\mapsto\big(\del_{\nu_{z}}\mathrm{Z}^k_{(\GDir, \GNeu)}(\xS,z) \cdot \del_{\nu_{z}}\mathrm{Z}^k_{(\GDir, \GNeu)}(y,z)\big)\in\RR$. 
\ElsIf{$\mathrm{Z}^k_{(\GDir, \GNeu)}(\xS,y) < 0$}
	\State Let $\mathsf{S}$ be the location of a global maxima of the function $\del\Om\ni z\mapsto\big(\del_{\nu_{z}}\mathrm{Z}^k_{(\GDir, \GNeu)}(\xS,z) \cdot \del_{\nu_{z}}\mathrm{Z}^k_{(\GDir, \GNeu)}(y,z)\big)\in\RR$. 
\EndIf

\State $(\GDir^0,\GNeu^0)\DEF (\GDir,\GNeu)$.
\While{$\textrm{True}$}
	\State Define $\GDel$ to be a boundary interval of length $2\eps$ with center $\mathsf{S}$.
	\State $(\GDir,\GNeu)\DEF (\GDir\setminus\overline{\GDel},\;\GNeu\cupdot\GDel)$.
	\State Compute the perturbed characteristic value $k$ associated to the partition $(\GDir,\GNeu)$ as described in Section \ref{subsection:EValDiffwithGenArgPrinc} or with the procedure given in \cite{Nigam2014}.
	\If{$|k-k_\star|\leq C_{\mathrm{tol}}$}
		\Return $(\GDir,\GNeu)$ 
	\ElsIf{$k_\star+C_{\mathrm{tol}}<k$}
		\State $\mathrm{BREAK}\;\mathrm{WHILE}$
	\Else
		\State $(\GDir,\GNeu)\DEF(\GDir^0,\GNeu^0)$ 
		\State $\eps\DEF\frac{\eps}{\sqrt{2}}$ 
	\EndIf
\EndWhile

\State $(\GDir^0,\GNeu^0)\DEF (\GDir,\GNeu)$
\While{$\textrm{True}$}
	\State Define $\GDel$ to be the extension of the Neumann interval boundary with center $j$, extended on both sides by $\eps/2$.
	\State $(\GDir,\GNeu)\DEF (\GDir\setminus\overline{\GDel},\;\GNeu\cup\GDel)$
	\State Compute the perturbed characteristic value $k$ associated to the partition $(\GDir,\GNeu)$ as described in Section \ref{subsection:EValDiffwithGenArgPrinc}.
	\If{$|k-k_\star|\leq C_{\mathrm{tol}}$}
		\Return $(\GDir,\GNeu)$ 
	\ElsIf{$k_\star+C_{\mathrm{tol}}<k$}
		\State $(\GDir^0,\GNeu^0)\DEF (\GDir,\GNeu)$
	\Else
		\State $(\GDir,\GNeu)\DEF(\GDir^0,\GNeu^0)$
		\State $\eps\DEF{\eps}\cdot0.9$
	\EndIf
\EndWhile
\end{algorithmic} 
\end{breakablealgorithm}

In the following we give an explanation for the choices.
\begin{itemize}
	\item[Line 2:] The reason we search for the next higher Dirichlet eigenvalue originates from the fact that, according to Proposition \ref{prop:lambda<lambda'}, when we insert Neumann boundaries, the corresponding eigenvalue decreases. The search for the next higher Dirichlet characteristic value an its multiplicity might be computationally expensive.
	\item[Line 3:] Using the algorithm proposed in \cite{EldarBruno2018}, we compute the Zaremba function using the decomposition $\mathrm{Z}^k(\xS,y)= \Gk(\xS, y) + \mathrm{R}^k(\xS, y)$, where $\Gk$ is the fundamental solution to the Helmholtz equation, and $\mathrm{R}^k$ satisfies the partial differential equation (\ref{pde:ZarembaRemainder}).
	More exactly, we obtain a function $\phi_{\mathrm{R}}$ on $\del\Om$, which is of the form in Proposition \ref{prop:asympt for u and phi}, with
	\begin{align*}
		\mathrm{R}^k(y) = \int_{\del\Om} \Gk(y,z)\phi_{\mathrm{R}}(z) \intd \sigma_z\,,
	\end{align*}
	for $y\in\Om$. Using the jump relations, see \cite[Section 2.3.2]{LPTSA}, we get for $y\rightarrow\GDir$ that
	\begin{align*}
		\del_{\nu_y}\mathrm{R}^k(y) 
			= \bigg( -\frac{1}{2}\mathrm{I}_{\del\Om}+\KGGkstar\bigg)\big[ \phi_{\mathrm{R}} \big](y)\,,
	\end{align*}
	where $\mathrm{I}_{\del\Om}$ denotes the identity operator.
	
	Using a discretization to the operator $\KGGkstar$, which we also readily obtain from \cite{Nigam2014}, we can calculate $\del_{\nu_y}\mathrm{Z}^k(\xS,\cdot)= \del_{\nu_y}\Gk(\xS, \cdot) + \del_{\nu_y}\mathrm{R}^k(\xS, \cdot)$.
	\item[Line 4-8:] In view of Lemma \ref{lemma:ZN-ZD}, we obtain that if $\mathrm{Z}^k_{(\GDir, \GNeu)}(\xS,y)\geq 0$ then we need a negative value of $\del_{\nu_{z}}\mathrm{Z}^k_{(\GDir, \GNeu)}(\xS,z) \cdot \del_{\nu_{z}}\mathrm{Z}^k_{(\GDir, \GNeu)}(y,z)$ to increase $\mathrm{Z}^k_{(\GDir, \GNeu)}(\xS,y)$ and vice-verca for $\mathrm{Z}^k_{(\GDir, \GNeu)}(\xS,y)\leq 0$. Taking the minima, respectively the maxima, we increase the absolute value of $\mathrm{Z}^k_{(\GDir, \GNeu)}(\xS,y)$. 
	
	We note that Lemma \ref{lemma:ZN-ZD} only holds for the case where $\Om$ is the unit circle, but we assume that it holds for all domains with smooth boundaries. We think, this can be established expanding the operator in \cite[Theorem 5.4]{HRMetasurfaceOnArxiv}.
	 
	From Theorem \ref{thm:ZarembaDecomp} we know that the Zaremba function is real valued, but due to numerical cancellation errors, the Zaremba function might have a non-zero imaginary part.
	
	In our numerical experiments, it always holds that a global minima is negative and a global maxima is positive, respectively. But we do not know if this holds true in general.
	\item[Line 10:] In this \textit{while}-loop we change a boundary interval with center $\mathsf{S}$ and length $2\eps$ into a Neumann Boundary condition. Then we compute an approximation $k$ to the new characteristic value. If $|k-k_\star|<C_{\mathrm{tol}}$, we end the algorithm, if $k+C_{\mathrm{tol}}<k_\star$, we break the \textit{while}-loop, and in the remaining case we decrease $\epsilon$ and go through the loop again.
	\item[Line 13:] To compute an approximation to the new characteristic value, which is smaller than $k$, we use the approximation stated in Proposition \ref{prop:keps-ko good approx}. To this end, we use as the complex domain $V$ encircling $k$ and $k_\star$ an ellipse with center $(k+k_\star)/2$ and semi-major axis $(k-k_\star)\cdot0.55$ and semi-minor axis $(k-k_\star)\cdot0.1$, which is to avoid complex characteristic values. Those factors are chosen due to good numerical results. A discretization to the operator $\calAk$ is computed using the algorithm described in \cite{Nigam2014}. For the complex derivative of $\calAk$, we used the rough approximation $(\calA(w+0.01)-\calA(w))/0.01$. The integral is approximated with a inbuilt-process.
	The approximation may yield the same result as the former characteristic value, that is $k$. In that case, the new characteristic value is not within $V$, which happens when the new boundary interval with Neumann boundary conditions is too long, or cannot be detected by the approximation.
	
	Here it might very well be that $k$ is not a simple eigenvalue, but instead for example a double eigenvalue, which occurs for $\Om$ being the unit circle. Then we search for both new eigenvalues and pick the one closer to $k_\star$, but still larger than $k_\star$. This search costs more time than the approximation algorithm.
	
	In numerical experiments it seems that the two eigenvalues of the double Dirichlet eigenvalue split such that one eigenvalue escapes subjectively faster from the double Dirichlet eigenvalue the longer the new boundary interval $\GDel$ is and the other eigenvalue subjectively slower. This is reminiscent of the behavior of the perturbation of a double eigenvalue in \cite{Dabrovski2017}, where the eigenvalue splits in an eigenvalue with difference $\OO(\eps^2)$ and an eigenvalue with difference $\OO(1/|\log(\eps)|)$, where $\eps$ is a value associated to the perturbation.
	\item[Line 23:] Next, we expand the boundary interval, which we established in Line 10-21. We expand it on both ends by a length $\eps/2$, whose factor $1/2$ is again chosen due to good numerical approval for minimizing runtime. Then we compute an approximation $k$ to the new characteristic value. If $|k-k_\star|<C_{\mathrm{tol}}$, we end the algorithm, if $k+C_{\mathrm{tol}}<k_\star$, we extend the boundary interval once again, else decrease $\epsilon$.
	\item[Line 26:] To compute an approximation to the new characteristic value, we use the same setting as in Line 13: The complex domain $V$ encircling $k$ and $k_\star$ is an ellipse with center $(k+k_\star)/2$ and semi-major axis $(k-k_\star)\cdot0.55$ and semi-minor axis $(k-k_\star)\cdot0.1$. A discretization to the operator $\calAk$ is computed using the algorithm described in \cite{Nigam2014}. For the complex derivative of said operator we used the rough approximation $(\calA(w+0.01)-\calA(w))/0.01$. The integral is approximated with a inbuilt-process.
	
	The approximation may again yield the same result as the former characteristic value, that is $k$, this happens when $\GDel$ is too long.
	
	In this \textit{while}-loop, it never happened that $k$ is not a simple eigenvalue. 
\end{itemize}

\begin{remark}
	If the function $\del\Om\ni z\mapsto\big(\del_{\nu_{z}}\mathrm{Z}^k_{(\GDir, \GNeu)}(\xS,z) \cdot \del_{\nu_{z}}\mathrm{Z}^k_{(\GDir, \GNeu)}(y,z)\big)\in\RR$ oscillates strongly on the boundary it might yield better results, when multiple, but smaller, boundary intervals are applied. The thought behind this is that using one long boundary interval might intersect the disadvantageous part of the function $\del_{\nu_{z}}\mathrm{Z}^k_{(\GDir, \GNeu)}(\xS,z) \cdot \del_{\nu_{z}}\mathrm{Z}^k_{(\GDir, \GNeu)}(y,z)$ and thus decrease the intensity of $\mathrm{Z}^k_{(\GDir, \GNeu)}(\xS,y)$.
\end{remark}

%
%
%
\newcommand{\raystretch}[1]{\renewcommand{\arraystretch}{#1}}

\section{Numerical Implementation and Tests}\label{Ch:NumImplTest}
Our first numerical test shows the algorithm in the best case scenario. We have the domain $\Omega=\{x\in\RR^2 \mid \NORM{x}_{\RR^2}<1\}$, the signal point $\xS = (0,0)^\mathrm{T}$,  the target characteristic value $k_\star=1$ and $C_\mathrm{tol}=10^{-3}$ and $\eps=0.1$. We remark here that the next higher Dirichlet characteristic value is a simple one at approximately $2.40482$. We let the receiving point $y\in\{ (0,r)^\mathrm{T}\in\RR^2 \mid r > 0\}$ vary. Here we want to mention that our implementation of the Zaremba function, as described in Section \ref{Ch:Algorithm}, comment on Line 3, yields a non-zero imaginary part for the Zaremba function, the same holds true for the approximation to the characteristic value $k$ as described in Section \ref{Ch:Algorithm}, comment on Line 13. We always choose the real part whenever in question. The number of discretization points for the operator $\calAk$ was $3\cdot64$. The results are displayed in Table \ref{table:1}. The Zaremba functions with Dirichlet boundary conditions and with final mixed boundary conditions, for the case $y= (0,0.5)^\mathrm{T}$, are displayed in Figure \ref{fig:Circ1}.  

\begin{figure}[h]
  \begin{subfigure}{0.49\textwidth}
    \centering
    \includegraphics[width=0.99\textwidth]{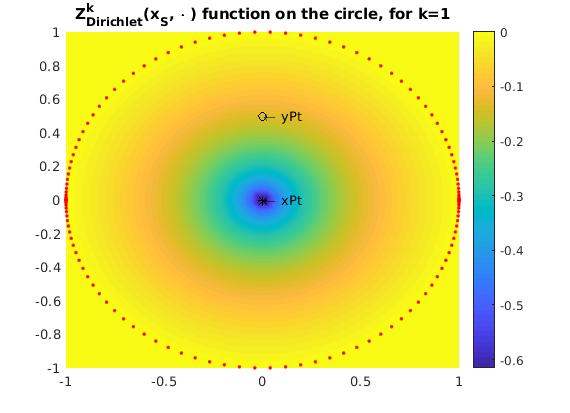}
  \end{subfigure}\hfill 
  \begin{subfigure}{0.49\textwidth} 
    \centering
    \includegraphics[width=0.99\textwidth]{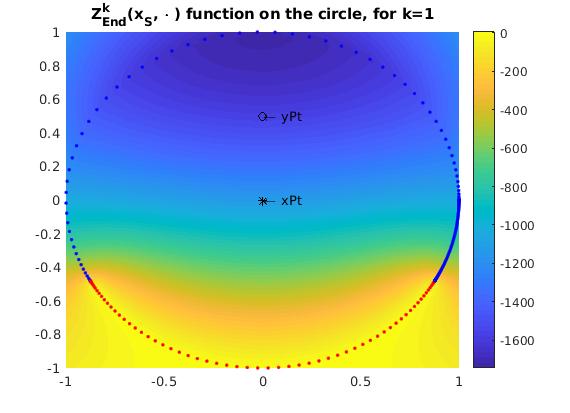}
  \end{subfigure}
  \caption{The Zaremba function for $k_\star=1$ on the unit disk with Dirichlet boundary condition on the left and final mixed boundary conditions on the right. Marked are $\xS$, denoted as 'xPt', and $y$, denoted as 'yPt'. The points on the boundary are our discretization points. Blue points denote the Neumann boundary conditions, red points denote the Dirichlet boundary conditions. \label{fig:Circ1}}
\end{figure}

\begin{figure}[h]
  \begin{subfigure}{0.49\textwidth}
    \centering
    \includegraphics[width=0.99\textwidth]{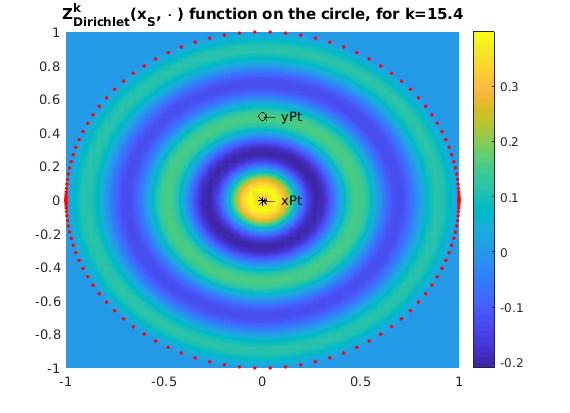}
  \end{subfigure}\hfill 
  \begin{subfigure}{0.49\textwidth} 
    \centering
    \includegraphics[width=0.99\textwidth]{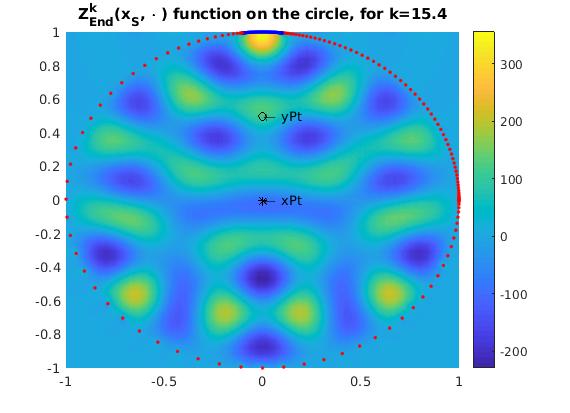}
  \end{subfigure}
  \caption{The Zaremba function for $k_\star=15.4$ on the unit disk with Dirichlet boundary condition on the left and final mixed boundary conditions on the right. Further notation is as in Figure \ref{fig:Circ1}. \label{fig:Circ2}}
\end{figure}

%
%
\begin{table*}\centering
\raystretch{1.1}
\resizebox{\columnwidth}{!}{%
\begin{tabular}{@{}rrrrrr@{}}
	& $r=0.1$ & $r=0.25$ & $r=0.5$ & $r=0.75$ & $r=0.9$\\ 
	\midrule
$\mathrm{Z}_{\text{Dirichlet}}^{k_\star}(\xS,y)$ 
	& -0.412
	& -0.261
	& -0.138 
	& -0.059 
	& -0.022 \\
$\mathrm{Z}_{\text{End}}^{k_\star}(\xS,y)$ 
	& -1288
	& -1438
	& -1634
	& -1754
	& -1788 \\
$\left|\frac{\mathrm{Z}_{\text{End}}^{k_\star}(\xS,y)}{\mathrm{Z}_{\text{Dirichlet}}^{k_\star}(\xS,y)}\right|$ 
	& 3123
	& 5503
	& 11824
	& 29623
	& 81687 \\
$\theta_{\mathrm{center}}$ 
	& $0.50 \pi$
	& $0.50 \pi$
	& $0.50 \pi$
	& $0.50 \pi$
	& $0.50 \pi$ \\
$l_{\mathrm{N}}$ 
	& $1.32 \pi$
	& $1.32 \pi$
	& $1.32 \pi$
	& $1.32 \pi$
	& $1.32 \pi$ \\
\bottomrule
\end{tabular}
}
\caption{We see Algorithm \ref{Algorithm} performing on the unit circle with $k_\star = 1$, $\xS = (0,0)^\mathrm{T}$, $y\in\{ (0,r)^\mathrm{T}\in\RR^2 \mid r > 0\}$, $C_\mathrm{tol}=10^{-3}$ and $\eps=0.1$. $\mathrm{Z}_{\mathrm{Dirichlet}}^k(\xS,y)$ represent the Zaremba function on the partition $(\del\Om,\varnothing)$ of the boundary and $\mathrm{Z}_{\mathrm{End}}^k(\xS,y)$ represents the Zaremba function on the final partition, where the final partition is made out of two boundary intervals, one with Dirichlet boundary conditions and the other with Neumann boundary conditions. $\theta_{\mathrm{center}}\in [0, 2\cdot pi)$ represents the angle of the center of the Neumann boundary intervals and $l_{\mathrm{N}}$ its length. The shown values are the real, rounded values of the numerical results.}\label{table:1}
\end{table*}

Our second numerical test shows the algorithm for a higher target characteristic value $k_\star$, namely $k_\star = 15.4$. We have as the domain $\Omega$ the unit circle $\{x\in\RR^2 \mid \NORM{x}_{\RR^2}<1\}$, as the signal point $\xS = (0,0)^\mathrm{T}$ and $C_\mathrm{tol}=10^{-3}$ and $\eps=0.05$. We remark here that the next higher Dirichlet characteristic value has multiplicity two and is at approximately $15.5898$. We let the receiving point $y\in\{ (0,r)^\mathrm{T}\in\RR^2 \mid r > 0\}$ vary. The number of discretization points for the operator $\calAk$ is $4\cdot48$. The results are displayed in Table \ref{table:2}. The Zaremba functions with Dirichlet boundary conditions and with final mixed boundary conditions, for the case $y= (0,0.5)^\mathrm{T}$, are displayed in Figure \ref{fig:Circ2}.  


\begin{table*}\centering
\raystretch{1.1}
\resizebox{\columnwidth}{!}{%
\begin{tabular}{@{}rrrrrr@{}}
	& $r=0.1$ & $r=0.25$ & $r=0.5$ & $r=0.75$ & $r=0.9$\\ 
	\midrule
$\mathrm{Z}_{\text{Dirichlet}}^{k_\star}(\xS,y)$ 
	& 0.341
	& -0.188 
	& 0.157 
	& -0.085 
	& 0.118 \\
$\mathrm{Z}_{\text{End}}^{k_\star}(\xS,y)$ 
	& 36.341
	& -14.271
	& 116.08
	& -15.811
	& 232.28 \\
$\left|\frac{\mathrm{Z}_{\text{End}}^{k_\star}(\xS,y)}{\mathrm{Z}_{\text{Dirichlet}}^{k_\star}(\xS,y)}\right|$ 
	& 106.6
	& 76.09
	& 739.0
	& 186.8
	& 1962 \\
$\theta_{\mathrm{center}}$ 
	& $0.50 \pi$
	& $1.90 \pi$
	& $0.50 \pi$
	& $0.46 \pi$
	& $0.50 \pi$ \\
$l_{\mathrm{N}}$ 
	& $0.064 \pi$
	& $0.064 \pi$
	& $0.064 \pi$
	& $0.064 \pi$
	& $0.064 \pi$ \\
\bottomrule
\end{tabular}
}
\caption{We see Algorithm \ref{Algorithm} performing on the unit circle with $k_\star = 15.4$, $\xS = (0,0)^\mathrm{T}$, $y\in\{ (0,r)^\mathrm{T}\in\RR^2 \mid r > 0\}$, $C_\mathrm{tol}=10^{-3}$ and $\eps=0.05$. $\mathrm{Z}_{\mathrm{Dirichlet}}^k(\xS,y)$, $\mathrm{Z}_{\mathrm{End}}^k(\xS,y)$, $\theta_{\mathrm{center}}$, and $l_{\mathrm{N}}$ are defined as in Table \ref{table:1}. The shown values are the real, rounded values of the numerical results.}\label{table:2}
\end{table*}

Our third numerical test shows the algorithm for a different domain $\Om$ namely a kite-shaped domain given by the following description for its boundary
\begin{align*}
	\begin{bmatrix}
		\cos(\tau)+0.65\cdot\cos(2\cdot\tau)-0.65 \\
     	1.5\cdot\sin(\tau)
	\end{bmatrix}\,,
\end{align*}
for $\tau\in[0, 2\pi)$. The target characteristic value is $k_\star = 1.5$. The signal point $\xS = (-1.25,1.25)^\mathrm{T}$ and receiving point $y= (-1.25,-1.25)^\mathrm{T}$. $C_\mathrm{tol}=10^{-2}$ and $\eps=0.05$. We remark here that the next higher Dirichlet characteristic value has multiplicity one and is at approximately $2.2099$. The number of discretization points for the operator $\calAk$ is $4\cdot48$. The result is displayed in Figure \ref{fig:Kite1}. 
The center of the Neumann boundary condition $\GNeu$ is at $(-1.191,-1.493)^\mathrm{T}$ with length $\approx 3.119$. $\mathrm{Z}_{\text{Dirichlet}}^k(\xS,y)\approx -4.05\,\cdot\,10^{-5}$ and $\mathrm{Z}_{\text{End}}^k(\xS,y)\approx -39.38$

In Figure \ref{fig:Kite2}, we have the same set-up but for $k_\star = 11.5$, with the next higher Dirichlet characteristic value around $11.6507$. Here, the center of the Neumann boundary condition $\GNeu$ is at $(-1.142, 0.641)^\mathrm{T}$ with length $\approx 0.632$. $\mathrm{Z}_{\text{Dirichlet}}^k(\xS,y)\approx 0.148$ and $\mathrm{Z}_{\text{End}}^k(\xS,y)\approx 1.68$.

\begin{figure}[h]
  \begin{subfigure}{0.49\textwidth}
    \centering
    \includegraphics[width=0.99\textwidth]{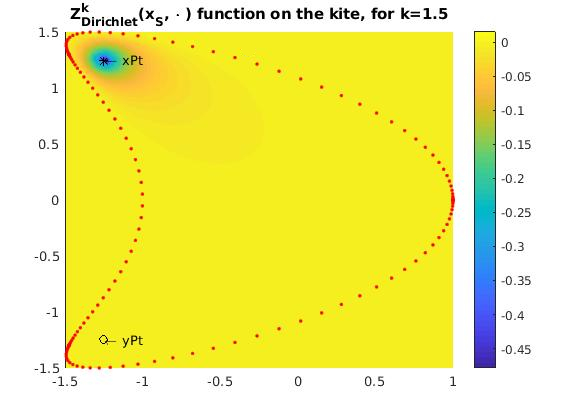}
  \end{subfigure}\hfill 
  \begin{subfigure}{0.49\textwidth} 
    \centering
    \includegraphics[width=0.99\textwidth]{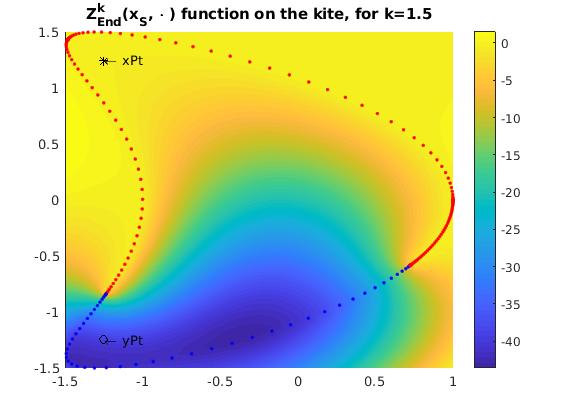}
  \end{subfigure}
  \caption{The Zaremba function for $k_\star=2$ on the kite shape with Dirichlet boundary condition on the left and final mixed boundary conditions on the right. Further notation is as in Figure \ref{fig:Circ1}. }\label{fig:Kite1}
\end{figure}

\begin{figure}[h]
  \begin{subfigure}{0.49\textwidth}
    \centering
    \includegraphics[width=0.99\textwidth]{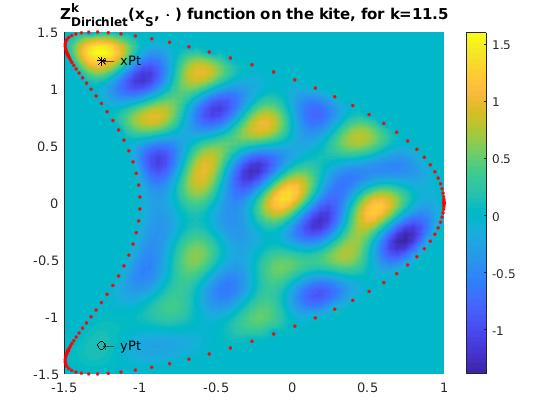}
  \end{subfigure}\hfill 
  \begin{subfigure}{0.49\textwidth} 
    \centering
    \includegraphics[width=0.99\textwidth]{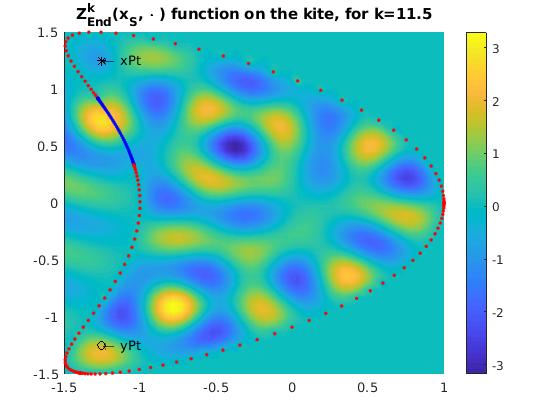}
  \end{subfigure}
  \caption{The Zaremba function for $k_\star=11.5$ on the kite shape with Dirichlet boundary condition on the left and final mixed boundary conditions on the right. Further notation is as in Figure \ref{fig:Circ1}.}\label{fig:Kite2}
\end{figure}



\bibliographystyle{plain}
\bibliography{refs}

\end{document}